%% file: main.tex
\documentclass{new_tlp}
\usepackage{mathptmx}
\DeclareMathAlphabet{\mathcal}{OMS}{cmsy}{m}{n}
\usepackage{amsmath}
\usepackage{amssymb}
\usepackage{color}
\usepackage{IEEEtrantools}
\usepackage{multicol}
\usepackage{enumitem}
\newcommand\bcmdtab{\noindent\bgroup\tabcolsep=0pt%
  \begin{tabular}{@{}p{10pc}@{}p{20pc}@{}}}
\newcommand\ecmdtab{\end{tabular}\egroup}

\def\At{\mathit{At}}
\newcommand{\eqdef}{\mathbin{\stackrel{\mathrm{def}}{=}}}

\newcommand\tuple[1]{\langle #1 \rangle}
\def\sneg{\sim\!\!}

\def\falsif{=\!\!\!\!| \;}
\newcommand\I[1]{\mathit{#1}}

\def\N5{{\cal N}_5}
\def\X5{{\cal X}_5}

\def\qed{\hfill$\Box$}
\newtheorem{theorem}{Theorem}
\newtheorem{definition}{Definition}
\newtheorem{example}{Example}
\newtheorem{proposition}{Proposition}
\newtheorem{lemma}{Lemma}
\newtheorem{corollary}{Corollary}
\newtheorem{observation}{Observation}

\newcommand{\set}[1]{\ensuremath{\{#1\}}}
\newcommand{\setb}[1]{\ensuremath{\big\{\ #1\ \big\}}}
\newcommand{\setB}[1]{\ensuremath{\Big\{\ #1\ \Big\}}}
\newcommand{\setm}[2]{\ensuremath{\{\ #1\ |\ #2\ \}}}
\newcommand{\setbm}[2]{\ensuremath{\big\{\ #1\ \ \big|\ \ #2\ \big\}}}
\newcommand{\setBm}[2]{\ensuremath{\Big\{\ #1\ \ \Big|\ \ #2\ \Big\}}}
\newcommand{\sset}[1]{\ensuremath{[#1]}}
\newcommand{\ssetb}[1]{\ensuremath{\big[\,#1\,\big]}}
\newcommand{\ssetB}[1]{\ensuremath{\Big[\,#1\,\Big]}}

\title[Revisiting Explicit Negation in Answer Set Programming]
        {Revisiting Explicit Negation in Answer Set Programming\thanks{This work was partially supported by MINECO, Spain, grant \mbox{TIC2017-84453-P}, Xunta de Galicia, Spain (GPC ED431B 2019/03 and 2016-2019 ED431G/01, CITIC). The third author is funded by the Centre International de Math\'{e}matiques et d'Informatique de Toulouse (CIMI) through contract ANR-11-LABEX-0040-CIMI within the programme ANR-11-IDEX-0002-02 and the Alexander von Humboldt Foundation..}}
\author[F. Aguado, P. Cabalar, J. Fandinno, D. Pearce, G. P{\'e}rez, C. Vidal]{
    FELICIDAD AGUADO$^1$, PEDRO CABALAR$^1$, JORGE FANDINNO$^2$ \and DAVID PEARCE$^3$, GILBERTO P{\'E}REZ$^1$, CONCEPCI{\'O}N VIDAL$^1$\\
    \\
    $^1$ Information Retrieval Lab, Centro de Investigaci\'on en Tecnolox\'ias da Informaci\'on e as Comunicaci\'ons (CITIC),\\ Universidade da Coru\~na, Spain\\
    \email{\{aguado,cabalar,gperez,eicovima\}@udc.es}\\ \\
    $^2$ IRIT, University of Toulouse, CNRS, France\\
	\email{jorge.fandinno@irit.fr}\\
	Universit\"{a}t Potsdam, Germany\\
    \email{fandinno@uni-potsdam.de}\\ \\
    $^3$ Universidad Polit{\'e}cnica de Madrid, Spain\\
	\email{david.pearce@upm.es}
}

\jdate{April 2019}
\pubyear{2019}
\pagerange{\pageref{firstpage}--\pageref{lastpage}}

\ifdefined\WITHREVIEWS %%%%%%%%%%%%%%%%
\definecolor{darkred}{rgb}{0.5,0.0,0.1}
\newcommand{\review}[2]{{\color{darkred} #1}\marginpar{\footnotesize {\color{darkred} #2}}}

\newcommand{\rev}[1]{{\color{blue} #1}}
\else
\definecolor{darkred}{rgb}{0.0,0.0,0.0}
\newcommand{\review}[2]{#1}

\fi
\begin{document}

%%%%%%%%%%%%%%%%%%%%%%%%%%%%%%%%%%%%%%%%%%%%%%%%%%%%%%%%%%%%%%%%%%%
\ifdefined\WITHREVIEWS
\input{letter}

\newpage
\pagestyle{plain}
\setcounter{page}{1}
\fi
%%%%%%%%%%%%%%%%%%%%%%%%%%%%%%%%%%%%%%%%%%%%%%%%%%%%%%%%%%%%%%%%%%%

\label{firstpage}

\maketitle

\begin{abstract}
A common feature in Answer Set Programming is the use of a second negation, stronger than default negation and sometimes called explicit, strong or classical negation.
This explicit negation is normally used in front of atoms, rather than allowing its use as a regular operator.
In this paper we consider the arbitrary combination of explicit negation with nested expressions, as those defined by Lifschitz, Tang and Turner.
We extend the concept of reduct for this new syntax and then prove that it can be captured by an extension of Equilibrium Logic with this second negation.
We study some properties of this variant and compare to the already known combination of Equilibrium Logic with Nelson's strong negation.
\emph{Under consideration for acceptance in TPLP.}
\end{abstract}


\section{Introduction}

Although the introduction of \emph{stable models}~\cite{GL88} in logic programming was motivated by the search of a suitable semantics for default negation, their early application to knowledge representation revealed the need of a second negation to represent explicit falsity.
This second negation was already proposed in~\cite{GelfondL91} under the name of \emph{classical negation}, an operator only applicable on atoms that, when present in the syntax, led to a change in the name of stable models to become \emph{answer sets}.
Classical negation soon became common in applications for commonsense reasoning and action theories~\cite{GL93} and was also extrapolated to the Well-Founded Semantics~\cite{Per92} under the name of \emph{explicit negation}.
Later on, it was incorporated to the paradigm of \emph{Answer Set Programming}~\cite{Nie99,MT99} (ASP), being nowadays present in the input language of most ASP solvers.

\review{
To understand the difference for knowledge representation between default negation (in this paper, written as $\neg$) and explicit negation (represented as $\sneg${}\ ), a typical example is to distinguish the rule
$\neg \I{train} \to \I{cross}$, that captures the criterion ``you can cross if you have no information on a train coming,'' from the (safier) encoding $\sneg \I{train} \to \I{cross}$ that means ``you can cross if you have evidence that no train is coming.'' In ASP, this explicit negation can only be used in front of atoms\footnote{In fact, the construct ``$\sneg \I{train}$'' is normally treated in ASP as a new atom $\I{train}'$ and an implicit constraint $\I{train} \wedge \I{train}' \to \bot$ is used to guarantee that both atoms cannot be true simultaneously.} so it is not seen as a real connective.
In an attempt of providing more flexibility to logic program connectives, \citeN{LTT99} introduced programs with \emph{nested expressions} where conjunction, disjunction and default negation could be arbitrarily nested both in the heads and bodies of rules, but classical negation was still restricted to an application on atoms.
To see an example, suppose that a given moment, three trains should be crossing, and we have an alarm that fires if one of them is known to be missing.
Using nested expressions, we can rewrite the program:
\begin{eqnarray*}
\sneg \I{train}_1 & \to & \I{alarm} \\
\sneg \I{train}_2 & \to & \I{alarm} \\
\sneg \I{train}_3 & \to & \I{alarm}
\end{eqnarray*}
as a single rule with a disjunction in the body:
\begin{eqnarray*}
\sneg \I{train}_1 \vee \sneg \I{train}_2 \vee \sneg \I{train}_3 & \to & \I{alarm}
\end{eqnarray*}
but we cannot further apply De Morgan to rewrite the rule above as:
\begin{eqnarray*}
\sneg (\I{train}_1 \wedge \I{train}_2 \wedge \I{train}_3) & \to & \I{alarm}
\end{eqnarray*}
It is easy to imagine that providing a semantics for this kind of expressions would be interesting if we plan to jump from the propositional case to programs with variables and aggregates (where, for instance, the number of trains is some arbitrary value $n \geq 0$).
}{\ref{rev2.1},\ref{rev3.1}}

An important breakthrough that meant a purely logical treatment, was the characterisation of stable models in terms of \emph{Equilibrium Logic} proposed by~\citeN{Pearce96}.
This non-monotonic formalism is defined in terms of a models selection criterion on top of the (monotonic) intermediate logic of \emph{Here-and-There} (HT)~\cite{Hey30} and captures default negation $\neg \varphi$ as a derived operator in terms of implication $\varphi \to \bot$, as usual in intuitionistic logic.
The definition of Equilibrium Logic also included a second, constructive negation `$\sneg$' corresponding to Nelson's \emph{strong negation}~\cite{Nel49} for intermediate logics.
In the case of HT, this extension yields a five-valued logic called $\N5$ where, although `$\sneg$' can now be nested as the rest of connectives, there exists a reduction for shifting it 
in front of atoms, obtaining a \emph{negative normal form} (NNF).
Once in NNF, the obtained equilibrium models actually coincide with answer sets for the syntactic fragments of nested expressions~\cite{LTT99} or for regular programs~\cite{GL93}.
For this reason, most papers on Equilibrium Logic for ASP assumed a reduction to NNF from the very beginning, and little attention was paid to the behaviour of formulas in the scope of strong negation under a logic programming perspective.
There are, however, cases in which this behaviour is not aligned with the reduct-based understanding of nested expressions in ASP.
Take, for instance, the formula:
\begin{eqnarray}
\sneg\neg p \to p \label{f:snegneg}
\end{eqnarray}
Its NNF reduction removes the combination of negations $\sneg \neg$ and produces the tautological rule \mbox{$p \to p$} whose unique equilibrium model is $\emptyset$, i.e., neither $p$ nor $\sneg p$ hold.
However, if we start instead from the formula $\sneg\neg \neg \neg p \to p$, the NNF reduction removes again the first pair of negations producing the rule $\neg \neg p \to p$ with a second answer set $\{p\}$.
This illustrates that we cannot replace $\neg p$ by $\neg \neg \neg p$ in the scope of strong negation, even though they would produce the same effect in any reduct of the style of~\cite{LTT99} for nested expressions.

In this paper, we consider a different characterisation of `$\sneg$' \ in HT and Equilibrium Logic.
We call this variant \emph{explicit negation} to differentiate it from Nelson's strong negation.
To test its adequacy, we start generalising the definition of nested expression by introducing an arbitrary nesting of `$\sneg$', adapting the definitions of reduct and answer set from~\cite{LTT99} to that context.
After that, we prove that equilibrium models (with explicit negation) capture the answer sets for these extended nested expressions and, in fact, preserve the strong equivalences from~\cite{LTT99} even for arbitrary formulas (including implication).
We also prove several properties of HT with explicit negation and provide a reduction to NNF that produces a different effect from $\N5$ when applied on implications or default negation. 

The rest of the paper is organised as follows.
In the next section, we introduce the extended definition of answer sets for programs with nested expressions, where explicit negation can be arbitrarily combined both in the rule bodies and the rule heads.
In Section~\ref{sec:eqx}, we present Equilibrium Logic with explicit negation and in particular, its new monotonic basis, $\X5$, since the selection of equilibrium models is the same one as in~\cite{Pearce96}.
Section~\ref{sec:fiveval} provides a five-valued characterisation of $\X5$ and studies different types of equivalence relations, including variants of strong equivalence.
In Section~\ref{sec:related}, we briefly explain the main differences between explicit ($\X5$) and strong ($\N5$) negations.
Finally, Section~\ref{sec:conc} concludes the paper.

%%%%%%%%%%%%%%%%%%%%%%%%%%%%%%%%%%%%%%%%%%%%%%%%%
\section{Nested expressions with explicit negation}
\label{sec:nested}

We begin describing the syntax of nested expressions, starting from a set of atoms $\At$.
A \emph{nested expression} $F$ is defined with the following grammar:
$$
F ::= \top \mid \bot \mid p \mid F \vee F \mid F \wedge F \mid \neg F \mid \ \sneg F
$$
where $p$ is any atom $p\in \At$.
The two negations $\neg$ and $\sneg$ \ are respectively called \emph{default} and \emph{explicit} negation (the latter is also called \emph{classical} in the ASP literature).
An \emph{explicit literal} is either an atom $p$ or its explicit negation $\sneg p$.
A \emph{default literal} is either an explicit literal $A$ or its default negation $\neg A$.
Thus, given atom $p$, we can form the default literals $p, \sneg p, \neg p$ and $\neg\!\!\sneg p$.
As we can see, the main difference with respect to~\cite{LTT99} is that, in that case, the explicit negation\footnote{To be precise, \cite{LTT99} used a different notation and names for operators: $\wedge$, $\vee$ and $\neg$ were respectively denoted as comma, semicolon and `not' in~\cite{LTT99}, whereas explicit negation $\sneg$ was denoted as $\neg$ and called \emph{classical negation}.} operator $\sneg$ \ was only used for explicit literals, whereas in this definition, it can be arbitrarily nested.
For instance, $\sneg(p \vee \neg q)$ is a nested expression under this new definition, but it is not under~\cite{LTT99}.
A \emph{rule} is an implication of the form $F \to G$ where $F$ and $G$ are nested expressions respectively called the \emph{body }ÃÂ and the \emph{head} of the rule.
A rule of the form $\top \to G$ is sometimes abbreviated as $G$ and is further called a \emph{fact} if $G$ is an explicit literal.
A \emph{logic program} is a set of rules.
We say that a nested expression, a rule or a program is \emph{explicit} if it does not contain default negation.

A program rule $F \to G$ is said to be \emph{regular} if the body $F=B_1 \wedge \dots \wedge B_n$ is a conjunction of default literals and the head $G=H_1 \vee \dots \vee H_m$ is a disjunction of default literals.
In a regular rule, we allow an empty body $n=0$ and write $F=\top$ or an empty head $m=0$ and $G=\bot$ but not both.
A program is \emph{regular} if all its rules are regular.

An \emph{interpretation} is a set of explicit literals that is consistent, that is, it does not contain both $p$ and $\sneg p$ for any atom $p$.
We define when an interpretation $T$ \emph{satisfies} (resp. \emph{falsifies}) a nested expression $F$, written $T \models F$ (resp. $T \falsif F$) providing the following recursive conditions:
\[
\begin{array}{r@{\,}c@{\,}ll@{\hspace{40pt}}r@{\,}c@{\,}ll}
T & \models & \top & & 
T & \not\falsif & \top \\
T & \not\models & \bot & & 
T & \falsif & \bot \\
T & \models & p & \mbox{if } p \in T & 
T & \falsif & p & \mbox{if } \sneg p \in T \\
T & \models & \varphi \wedge \psi & \mbox{if } T\models \varphi \mbox{ and } T \models \psi& 
T & \falsif & \varphi \wedge \psi & \mbox{if } T\falsif \varphi \mbox{ or } T \falsif \psi \\
T & \models & \varphi \vee \psi & \mbox{if } T\models \varphi \mbox{ or } T \models \psi& 
T & \falsif & \varphi \vee \psi & \mbox{if } T\falsif \varphi \mbox{ and } T \falsif \psi \\
T & \models & \sneg \varphi & \mbox{if } T \falsif \varphi & 
T & \falsif & \sneg \varphi & \mbox{if } T \models \varphi\\
T & \models & \neg \varphi & \mbox{if } T \not\models \varphi & 
T & \falsif & \neg \varphi & \mbox{if } T \models \varphi
\end{array}
\]
As an example, given $\At=\{p,q\}$ and $T=\{\sneg p\}$ we have $T \models \sneg p \vee q$ because $T \models\, \sneg p$ (i.e. $T \falsif p$) although neither $T \models q$ nor $T \falsif q$, that is, $q$ is undefined.
The latter can be expressed as $T \models \neg q \wedge \neg\!\sneg q$ (i.e., $q$ is neither true nor false).
As another example, $T \falsif p \wedge q$ because $T \falsif p$ even though, as we said, $q$ is undefined.
We say that $\varphi$ is \emph{valid} if we have $T \models \varphi$ for every interpretation~$T$.
The logic induced by these valid expressions precisely corresponds to \emph{classical logic with strong negation} as studied by~\citeN{vakarelov1977notes}.
Note that, as usual in classical logic, $\varphi \to \psi$ is definable as $\neg\varphi \vee \psi$ in this context.

Let $\Pi$ be an explicit program.
A consistent set of literals $T$ is a \emph{model} of $\Pi$ if, for every rule $F \to G$ in $\Pi$, $T \models G$ whenever $T \models F$.

\begin{definition}[reduct]
The reduct of a nested expression $F$ with respect to an interpretation $T$ is denoted as $F^T$ and defined recursively as follows:
$$
\begin{array}{rcll}
p^T & \eqdef & p & \mbox{for any atom } p \in \At \\
(F \wedge G)^T & \eqdef & F^T \wedge G^T \\
(F \vee G)^T & \eqdef & F^T \vee G^T \\
(\sneg F)^T & \eqdef & \sneg (F^T)\\
(\neg F)^T & \eqdef & \left\{
\begin{array}{rl}
\bot & \mbox{if } T \models F \\ 
\top & \mbox{otherwise}
\end{array}
\right.\\
\end{array}
$$
The \emph{reduct} of a program $\Pi$ with respect to $T$ corresponds to the explicit program:\\
$\Pi^T \eqdef \{ \ (F^T \to G^T) \mid (F \to G) \in \Pi\ \}$.\qed
\end{definition}
\begin{proposition}
\label{prop:total_model_reduct}
For any consistent set of literals $T$ and any nested formula $F$:
\begin{itemize}[ leftmargin=15pt]
\item $T \models F$ iff $T \models F^T$;
\item $T \falsif F$ iff $T \falsif F^T$.
\qed
\end{itemize}
\end{proposition}

\begin{definition}[answer set]
A consistent set of literals $T$ is an \emph{answer set} of a program $\Pi$ if it is a $\subseteq$-minimal model of the reduct $\Pi^T$.\qed
\end{definition}

Notice that the definitions of reduct and answer set for the case of regular programs directly coincide with the standard definitions in ASP without nested expressions~\cite{GelfondL91}.
They also coincide with~\cite{LTT99}, defined on the case of programs with nested expressions where `$\sneg$' \ is only in front of atoms.

\begin{example}\label{ex:nonot}
Take the program consisting of the single rule \eqref{f:snegneg}. 
For $\At=\{p\}$, we have three possible interpretations $T_1=\{p\}$, $T_2=\{\sneg p\}$ and $T_3=\emptyset$.
This yields two possible reducts $\Pi^{T_1}=\{\sneg \bot \to p\}$ and $\Pi^{T_2}=\Pi^{T_3}=\{\sneg \top \to p\}$.
It is easy to see that their corresponding minimal models are $T_1$ and $T_3$ which constitute the two answer sets of $\Pi$. \qed
\end{example}

\begin{example}\label{ex:bird}
Take the program consisting of the single rule:
\begin{eqnarray}
\neg (\I{bird} \wedge \sneg \I{flies}) \to \ \sneg (\I{bird} \wedge \sneg \I{flies}) \label{f:bird}
\end{eqnarray}
capturing the idea that ``being a bird that does not fly'' should be false by default.
If we choose any interpretation $T$ such that $T \models \I{bird} \wedge \sneg \I{flies}$ then the reduct will have a single rule with $\bot$ in the body and the minimal model will be $\emptyset$ which does not  satisfy $\I{bird} \wedge \sneg \I{flies}$.
If $T \not\models \I{bird} \wedge \sneg \I{flies}$ instead, the reduct becomes $\top \to \ \sneg (\I{bird} \wedge \sneg \I{flies})$ and the minimal models of this program are $\{\sneg \I{bird}\}$ and $\{\I{flies}\}$ that, as they are both compatible with the assumption for $T$, they become the two answer sets of \eqref{f:bird}.

Suppose we extend now \eqref{f:bird} with the fact $bird$.
Doing so, it is easy to see that the only answer set becomes $\{\I{flies}\}$.
Analogously, if we take  \eqref{f:bird} plus the fact $\sneg \I{flies}$ the only answer set becomes $\{\sneg \I{bird}\}$.
Finally, if we add the facts $\I{bird}$ and $\sneg \I{flies}$ to \eqref{f:bird}, the default is deactivated and we get the unique answer set $\{\I{bird},\sneg \I{flies}\}$. \qed
\end{example}

%%%%%%%%%%%%%%%%%%%%%%%%%%%%%%%%%%%%%%
\section{Equilibrium logic with explicit negation}
\label{sec:eqx}

We start defining the monotonic logic of \emph{Here-and-There with explicit negation}, $\X5$.
Let $\At$ be a set of atoms.
A \emph{formula} $\varphi$ is an expression built with the grammar:
$$
\varphi ::= p \mid \bot \mid \varphi \wedge \varphi \mid \varphi \vee \varphi \mid \varphi \to \varphi \mid \ \sneg \varphi
$$
for any atom $p\in \At$.
We also use the abbreviations:
\setlength{\multicolsep}{-10pt}
\begin{multicols}{2}
\begin{eqnarray*}
\neg \varphi & \eqdef & (\varphi \to \bot)\\
\top & \eqdef & \neg \bot\\
\end{eqnarray*}
\begin{eqnarray*}
\\
\varphi \leftrightarrow \psi & \eqdef & (\varphi \to \psi) \wedge (\psi \to \varphi)\\
\varphi \Leftrightarrow \psi & \eqdef & (\varphi \leftrightarrow \psi) \wedge (\sneg \varphi \leftrightarrow \sneg \psi)
\end{eqnarray*}
\end{multicols}
% \vspace{5pt}
\noindent
Given a pair of formulas $\varphi$ and $\alpha$, we write $\varphi[\alpha/p]$ to denote the uniform substitution of all occurrences of atom $p$ in $\varphi$ by $\alpha$.
As usual, a \emph{theory} is a set of formulas.
We sometimes understand finite theories (or subtheories) as the conjunction of their formulas.
Notice that programs with nested expressions are also theories under this definition.

An $\X5$-\emph{interpretation} is a pair $\tuple{H,T}$ of consistent sets of explicit literals (respectively standing for ``here'' and ``there'') satisfying $H \subseteq T$.
We say that the interpretation is \emph{total} when $H=T$.
\begin{definition}[$\X5$ Satisfaction/falsification]\label{def:satfals}
We say that $\tuple{H,T}$ \emph{satisfies} (resp. \emph{falsifies}) a formula $\varphi$, written $\tuple{H,T} \models \varphi$ (resp. $\tuple{H,T} \falsif \varphi$), when the following recursive conditions hold:
\[
\begin{array}{r@{\,}c@{\,}l@{\;}l@{\hspace{10pt}}r@{\,}c@{\,}l@{\;}l}
\tuple{H,T} & \models & \top & & 
\tuple{H,T} & \not\falsif  & \top \\
\tuple{H,T} & \not\models  & \bot & & 
\tuple{H,T} & \falsif & \bot \\
\tuple{H,T} & \models & p & \mbox{if } p \in H & 
\tuple{H,T} & \falsif & p & \mbox{if } \sneg p \in H \\
\tuple{H,T} & \models & \varphi \wedge \psi & \mbox{if } \tuple{H,T} \models  \varphi \mbox{ and } \tuple{H,T}  \models  \psi& 
\tuple{H,T} & \falsif & \varphi \wedge \psi & \mbox{if } \tuple{H,T}  \falsif  \varphi \mbox{ or } \tuple{H,T}  \falsif  \psi \\
\tuple{H,T} & \models & \varphi \vee \psi & \mbox{if } \tuple{H,T} \models  \varphi \mbox{ or } \tuple{H,T}  \models  \psi& 
\tuple{H,T} & \falsif & \varphi \vee \psi & \mbox{if } \tuple{H,T}  \falsif  \varphi \mbox{ and } \tuple{H,T}  \falsif  \psi \\
\tuple{H,T} & \models & \sneg \varphi & \mbox{if } \tuple{H,T}  \falsif  \varphi & 
\tuple{H,T} & \falsif & \sneg \varphi & \mbox{if } \tuple{H,T}  \models  \varphi\\
\tuple{H,T} & \models & \varphi\! \to \! \psi & \mbox{if both} & 
\tuple{H,T} & \falsif & \varphi \! \to \! \psi & \mbox{if } \tuple{T,T}  \models  \varphi \mbox{ and } \tuple{H,T} \falsif \psi\\
& & & (i) \tuple{H,T} \not\models   \varphi \mbox{ or } \tuple{H,T} \models \psi \\
& & & (ii) \tuple{T,T} \not\models   \varphi \mbox{ or } \tuple{T,T} \models \psi
& & & & \hfill\Box
\end{array}
\]
\end{definition}
A formula $\varphi$ is a \emph{tautology} (or is \emph{valid}), written $\models \varphi$, if it is satisfied by every possible interpretation.
We say that an $\X5$-interpretation $\tuple{H,T}$ is a \emph{model} of a theory $\Gamma$, written $\tuple{H,T} \models \Gamma$, if $\tuple{H,T}\models \varphi$ for all $\varphi \in \Gamma$.
The next observation about Definition~\ref{def:satfals} connects satisfaction `$\models$' with standard HT.
\begin{observation}\label{obs:ht}
 The satisfaction relation `$\models$' (left column in Def.~\ref{def:satfals}) of any formula corresponds to regular HT satisfaction up to the first occurrence of `$\sim$', where the falsification `$\falsif$' comes into play.\qed
\end{observation}
As a result, any tautology from HT can be shifted to $\X5$, even if its atoms are uniformly replaced by subformulas containing explicit negation.
\begin{theorem}\label{th:httaut}
If formula $\varphi$ is HT valid (and so, it does not contain $\sneg$~) then $\varphi[\alpha/p]$ is also $\X5$ valid, for any formula $\alpha$ and any atom $p$.\qed
\end{theorem}
If we choose any $p$ not occurring in $\varphi$, then $\varphi[\alpha/p]=\varphi$ and the theorem above is just saying that $\X5$ is a conservative extension of HT.
But it can also be exploited further by replacing, in the HT tautology, any atom by an arbitrary formula containing negation.
For instance, if explicit negation only occurs in front of atoms, we essentially get HT with explicit literals playing the role of atoms (disregarding inconsistent models).
However, when we combine explicit negation in an arbitrary way, some usual properties of HT need to be checked in the new context.

\begin{lemma}
\label{lem:satisfaction_total_models}
Let $T$ be a consistent set of literals and $F$ a nested expression. Then:
\begin{itemize}[ leftmargin=15pt]
\item $\tuple{T,T} \models F$ iff $T \models F$;
\item $\tuple{T,T} \falsif F$ iff $T \falsif F$. \qed
\end{itemize}
\end{lemma}

\begin{theorem}[Persistence]
\label{th:persistence}
For any $\X5$-interpretation $\tuple{H,T}$ and any formula $\varphi$ then both:
\begin{itemize}[ leftmargin=15pt]
\item[(i)] $\tuple{H,T} \models \varphi$ implies $\tuple{T,T} \models \varphi$; \item[(ii)] $\tuple{H,T} \falsif \varphi$ implies $\tuple{T,T} \falsif \varphi$.\qed
\end{itemize}
\end{theorem}

\begin{proposition}
\label{prop:default_negation}
For any $\X5$-interpretation $\tuple{H,T}$, any formula $\varphi$:
\begin{itemize}[ leftmargin=15pt]
\item $\tuple{H,T} \models \neg \varphi$ iff $\tuple{T,T} \not\models \varphi$;
\item $\tuple{H,T} \falsif \neg \varphi$ iff $\tuple{T,T} \models \varphi$.
\qed
\end{itemize}
\end{proposition}

The following results establish a connection between $\X5$ and the reduct of a nested expression or a program.

\begin{lemma}
\label{lem:aux_reduct}
Let $\tuple{H,T}$ be an $\X5$-interpretation and $F$ a nested expression. Then:
\begin{itemize}[ leftmargin=15pt]
\item $\tuple{H,T} \models F$ iff $H \models F^T$;
\item $\tuple{H,T} \falsif F$ iff $H \falsif F^T$. \qed
\end{itemize}
\end{lemma}

\begin{corollary}
\label{cor:equivalence_for_total_model}
For any consistent set of literals $T$ and any program $\Pi$:
$\tuple{T,T} \models \Pi$ iff $T \models \Pi$.\qed
\end{corollary}

\begin{proposition}\label{prop:htreduct}
For any $\X5$-intepretation $\tuple{H,T}$ and any program $\Pi$: \begin{center} $\tuple{H,T} \models \Pi$ iff $H$ is a model of $\Pi^T$ and $T$ is a model of $\Pi$.
\end{center}
\vspace*{-23pt}\qed
\end{proposition}

\begin{definition}[Equilibrium model]\label{def:eqmodel}
A total $\X5$-interpretation $\tuple{T,T}$ is an \emph{equilibrium model} of a theory $\Gamma$ if $\tuple{T,T}$ is a model of $\Gamma$ and there is no other model $\tuple{H,T}$ of $\Gamma$ with $H \subset T$.\qed
\end{definition}

\emph{Equilibrium logic (with explicit negation)} is the non-monotonic logic induced by equilibrium models.
The following theorem guarantees that equilibrium models and answer sets coincide for the syntax of programs with nested expressions.
\begin{theorem}\label{th:answersets}
An interpretation $T$ is an answer set of a program $\Pi$ iff $\tuple{T,T}$ is an equilibrium model of $\Pi$.\qed
\end{theorem}

\review{To conclude this section, we provide an alternative reduct definition for arbitrary formulas (and not just nested expressions) obtained as a generalisation of Ferraris' reduct~\cite{Fer05}.
This generalisation introduces a main feature\footnote{We also provide a translation for implications $\alpha \to \beta$ but this is not strictly necessary: for computing the reduct, they can be previously replaced by $\neg \alpha \vee \beta$.} with respect to~\cite{Fer05}: it actually uses two dual transformations, $\varphi^T_+$ and $\varphi^T_-$, to obtain a symmetric behaviour depending on the number of explicit negations in the scope.

\begin{definition}\label{def:Ferraris_reduct}
Given a formula $\varphi$ and an interpretation $T$ (a consistent set of explicit literals) we define the following pair of mutually recursive transformations:
\[
\begin{array}{cc}
\varphi^T_+ \eqdef \left\{
\begin{array}{cl}
\bot & \text{if } T \not\models \varphi \\
p & \text{if } \varphi=p \in \At, p \in T \\
\alpha^T_+ \otimes \beta^T_+ & \text{if } T \models \varphi, \varphi=\alpha \otimes \beta, \\
& \text{ for } \otimes \in\{\vee,\wedge \}\\ 
\neg (\alpha^T_+) \vee \beta^T_+ & \text{if } T \models \varphi, \varphi=\alpha \to \beta \\
\neg (\alpha^T_+) & \text{if } T \models \varphi, \varphi=\neg \alpha,  \\
\sneg (\alpha^T_-) & \text{if } T \models \varphi, \varphi=\sneg \alpha
\end{array}
\right.
&
\varphi^T_- \eqdef \left\{
\begin{array}{cl}
\top & \text{if } T \not\falsif \varphi\\
p & \text{if } \varphi=p \in \At, \sneg p \in T \\
\alpha^T_- \otimes \beta^T_- & \text{if } T \falsif \varphi, \varphi=\alpha \otimes \beta, \\
& \text{ for } \otimes \in\{\vee,\wedge \}\\ 
\beta^T_- & \text{if } T \falsif \varphi, \varphi=\alpha \to \beta\\
\bot & \text{if } T \falsif \varphi, \varphi=\neg \alpha\\
\sneg (\alpha^T_+) & \text{if } T \falsif \varphi, \varphi=\sneg \alpha
\end{array}
\right.
\end{array}
\]
The reduct $\Gamma^T_+$ of a theory $\Gamma$ is just defined as the set $\{\varphi^T_+ \mid \varphi \in \Gamma\}$.\qed
\end{definition}
For instance, given  $\varphi=\eqref{f:bird}$ and $T=\{\sneg \I{bird}\}$, the reader can check that the application of the definition above eventually produces the formula $\varphi^T_+ = \neg \neg \bot \vee \sneg (\I{bird} \wedge \top)$ which is equivalent to $\sneg \I{bird}$.
If we take $T=\{\I{flies}\}$ instead, the result is $\varphi^T_+=\neg \neg \bot \vee \sneg (\top \wedge \sneg \I{flies})$ that is equivalent to $\I{flies}$.
As a third example, if we take $T=\{\I{bird}\}$ then we directly get $\varphi^T_+=\bot$.

\begin{theorem}\label{th:Ferraris_reduct}
For any formula $\varphi$ and any pair of interpretations $H \subseteq T$:
\begin{enumerate}
\item[](i) \ $H \models \varphi^T_+$ iff $\tuple{H,T} \models \varphi$;
\item[](ii) $H \falsif \varphi^T_-$ iff $\tuple{H,T} \falsif \varphi$.\qed
\end{enumerate}
\end{theorem}

From Lemma~\ref{lem:aux_reduct} and Theorem~\ref{th:Ferraris_reduct} we immediately conclude:

\begin{corollary}\label{cor:reducts}
For any nested expression $F$ and any pair of interpretations $H \subseteq T$:
\begin{enumerate}[ leftmargin=15pt]
\item[(i)] $H \models F^T$ iff $T\models F$ and $H \models F^T_+$;
\item[(ii)] $H \falsif F^T$ iff $T\falsif F$ and $H \falsif F^T_-$.\qed
\end{enumerate}
\end{corollary}

\begin{corollary}\label{cor:reducteq}
$\tuple{T,T}$ is an equilibrium model of $\Gamma$ iff $T$ is a minimal model of $\Gamma^T_+$.\qed
\end{corollary}
Back to the example formula $\varphi=$\eqref{f:bird}, taking $T=\{\sneg \I{bird}\}$ we saw that $\varphi^T_+$ is equivalent to $\sneg \I{bird}$ whose minimal model is obviously $T$.
Therefore, $\tuple{T,T}$ is an equilibrium model.}{\ref{rev1.1}}

%%%%%%%%%%%%%%%%%%%%%%%%%%%%%%%%%%%%%%%%%%%%%%%%%%%%%%%%%%%%
\section{Multivalued characterisation and equivalence relations}
\label{sec:fiveval}

An alternative way of characterising $\X5$ is as a five-valued logic defined as follows.
Given any $\X5$-interpretation $M=\tuple{H,T}$ we define its corresponding 5-valued mapping $M: \At \to \{-2,-1,0,1,2\}$ so that, for any atom $p\in \At$:
$$
M(p) \eqdef \left\{
\begin{array}{rl}
2 & \mbox{if } p \in H\\
-2 & \mbox{if } \sneg p \in H\\
1 & \mbox{if } p \in T\setminus H\\
-1 & \mbox{if } \sneg p \in T\setminus H\\
0 & \mbox{otherwise, i.e., } p \not\in T, \sneg p \not\in T
\end{array}
\right.
$$
We can read these five values as follows: $2$ = \emph{proved to be true}; $-2$ = \emph{proved to be false}; $1$ = \emph{true by default}; $-1$ = \emph{false by default}; and $0$ = \emph{undefined}.
Notice that values $1$ and $-1$ are used for explicit literals in $T \setminus H$.
As a consequence:
\begin{proposition}
An $\X5$-interpretation $M=\tuple{H,T}$ is total (i.e. $H=T$) iff $M(p) \in \{-2,0,2\}$ for all $p \in \At$.\qed
\end{proposition}

\begin{definition}[Valuation of formulas]\label{def:valuation}
This 5-valuation can be extended to arbitrary formulas in the following way:
\begin{IEEEeqnarray*}{lCl+x*}
M(\bot) & \eqdef & -2 \\
M(\top) & \eqdef & 2 \\
M(\varphi \wedge \psi) & \eqdef & \min(M(\varphi),M(\psi)) \\
M(\varphi \vee \psi) & \eqdef & \max(M(\varphi),M(\psi)) \\
M(\varphi \to \psi) & \eqdef & \left\{
\begin{array}{cl}
2 & \mbox{if } M(\varphi) \leq \max(M(\psi),0) \\
M(\psi) & \mbox{otherwise }
\end{array}
\right.\\
M(\sneg \varphi) & \eqdef & -M(\varphi)
&\qed
\end{IEEEeqnarray*}
\end{definition}
The designated value is $2$, that is, we will understand that a formula is satisfied when $M(\varphi)=2$.
Moreover, a complete correspondence with the satisfaction/falsification of formulas given in the previous section is fixed by the following theorem:
\begin{theorem}\label{th:corresp}
For any $\X5$-interpretation $M=\tuple{H,T}$ and any formula $\varphi$:
\setlength{\multicolsep}{5pt}
\begin{multicols}{2}
\begin{itemize}[ leftmargin=15pt]
\item $\tuple{H,T} \models \varphi$ iff $M(\varphi)=2$;
\item $\tuple{T,T} \models \varphi$ iff $M(\varphi)>0$;
\item $\tuple{H,T} \falsif \varphi$ iff $M(\varphi)=-2$;
\item $\tuple{T,T} \falsif \varphi$ iff $M(\varphi)<0$.\qed
\end{itemize}
\end{multicols}
\end{theorem}
The equilibrium condition given in Definition~\ref{def:eqmodel} can be rephrased in 5-valued terms as follows.
Given two $\X5$-interpretations $M=\tuple{H,T}$ and $M'=\tuple{H',T'}$ we say that $M$ is \emph{smaller} than $M'$, written $M \leq M'$, when $T=T'$ and $H \subseteq H'$.
\begin{proposition}\label{prop:leq}
Let $M$ and $M'$ be a pair of $\X5$-interpretations.
Then $M \leq M'$ iff, for any atom $p \in \At$, the following three conditions hold:
\begin{enumerate}[ leftmargin=15pt]
    \item $M(p)=0$ iff $M'(p)=0$;
    \item If $M(p) >0$, then $M(p) \leq M'(p)$;
    \item If $M(p) <0$, then $M'(p) \leq M(p)$.\qed
\end{enumerate}
\end{proposition}

\begin{theorem}
A total interpretation $M=\tuple{T,T}$ is an equilibrium model of a theory $\Gamma$ iff $M(\varphi)=2$ for all $\varphi \in \Gamma$ and there is no $M' < M$ such that $M'(\varphi)=2$ for all $\varphi \in \Gamma$.
\end{theorem}
\begin{proof}
It follows from Theorem~\ref{th:corresp} and the definition of $\leq$ relation.
\end{proof}

The truth tables derived from Definition~\ref{def:valuation} are depicted in Figure~\ref{fig:tables}, including the tables for derived operators `$\neg$', `$\leftrightarrow$' and `$\Leftrightarrow$'.
Note that the table for $\neg \varphi=(\varphi \to \bot)$ is just the first column of the table for `$\to$' since the evaluation of `$\bot$' is fixed to $-2$.
\begin{figure}[htbp]
\centering
$$
\begin{array}{c@{\hspace{5pt}}c}
\begin{array}{r|rrrrr}
\wedge & -2 & -1 & 0 & 1 & 2\\
\hline
-2 & -2 & -2 & -2 & -2 & -2 \\
-1 & -2 & -1 & -1 & -1 & -1 \\
0 & -2 & -1 & 0 & 0 & 0 \\
1 & -2 & -1 & 0 & 1 & 1 \\
2 & -2 & -1 & 0 & 1 & 2
\end{array}
&
\begin{array}{r|rrrrr}
\vee & -2 & -1 & 0 & 1 & 2\\
\hline
-2 & -2 & -1 & 0 & 1 & 2 \\
-1 & -1 & -1 & 0 & 1 & 2 \\
0 & 0 & 0 & 0 & 1 & 2 \\
1 & 1 & 1 & 1 & 1 & 2 \\
2 & 2 & 2 & 2 & 2 & 2
\end{array}
\\ \\
\begin{array}{r|rrrrr}
\to & -2 & -1 & 0 & 1 & 2\\
\hline
-2 & 2 & 2 & 2 & 2 & 2 \\
-1 & 2 & 2 & 2 & 2 & 2 \\
0 & 2 & 2 & 2 & 2 & 2 \\
1 & -2 & -1 & 0 & 2 & 2 \\
2 & -2 & -1 & 0 & 1 & 2
\end{array}
&
\begin{array}{c@{\hspace{20pt}}c}
\begin{array}{r|r}
\varphi & \sneg \varphi\\
\hline
-2 & 2 \\
-1 & 1 \\
0 & 0 \\
1 & -1 \\
2 & -2 
\end{array}
&
\begin{array}{r|r}
\varphi & \neg \varphi\\
\hline
-2 & 2 \\
-1 & 2 \\
0 & 2 \\
1 & -2 \\
2 & -2 
\end{array}
\end{array}
\\ \\
\begin{array}{r|rrrrr}
\leftrightarrow & -2 & -1 & 0 & 1 & 2\\
\hline
-2 & 2 & 2 & 2 & -2 & -2 \\
-1 & 2 & 2 & 2 & -1 & -1 \\
0 & 2 & 2 & 2 & 0 & 0 \\
1 & -2 & -1 & 0 & 2 & 1 \\
2 & -2 & -1 & 0 & 1 & 2
\end{array}
&
\begin{array}{r|rrrrr}
\Leftrightarrow & -2 & -1 & 0 & 1 & 2\\
\hline
-2 & 2 & 1 & 0 & -2 & -2 \\
-1 & 1 & 2 & 0 & -1 & -2 \\
0 & 0 & 0 & 2 & 0 & 0 \\
1 & -2 & -1 & 0 & 2 & 1 \\
2 & -2 & -2 & 0 & 1 & 2
\end{array}
\end{array}
$$
\caption{Truth tables for $\X5$.}
\label{fig:tables}
\end{figure}
It is easy to check, for instance, that the following implication is valid:
\begin{eqnarray}
\sneg \varphi \to \neg \varphi \label{f:coher}
\end{eqnarray}
expressing that explicit negation is stronger than default negation\footnote{This property is called the \emph{coherence} principle in~\cite{Per92}.}.
Moreover, default negation is definable in terms of implication and explicit negation (without resorting to $\bot$) since, with some effort, it can be checked that the table for $\neg \varphi$ can be equally obtained through the expression:
\begin{eqnarray*}
\sneg ((\varphi \to \ \sneg \varphi) \to \ \sneg (\varphi \to \ \sneg \varphi))
\end{eqnarray*}
An important remark regarding equivalence is that to express that this (or any) pair of formulas are equivalent, double implication does not suffice.
This is because, as we can see in the tables,
$M(\varphi \leftrightarrow \psi)=2$ does not imply that $M(\varphi)=M(\psi)$. 
To get such a correspondence, we must resort instead to the stronger `$\Leftrightarrow$' for which $M(\varphi \Leftrightarrow \psi)=2$ holds if and only if $M(\varphi)=M(\psi)$.
This lack of the `$\leftrightarrow$' equivalence (we call it \emph{weak} equivalence) has an important consequence: it does not define a congruence relation since $\models \alpha \leftrightarrow \beta$ no longer implies that we can freely replace subformula $\alpha$ by $\beta$ in any arbitrary context: it may be the case that $\not\models \sneg \alpha \leftrightarrow \sneg \beta$.
For instance, we can easily check that $\models p \wedge \neg p \leftrightarrow \bot$ because $\min(M(p),M(\neg p)) \leq 0$ and $M(\bot)=-2$, so $M(p \wedge \neg p \leftrightarrow \bot)=2$ for any $M$.
However, we cannot replace $p \wedge \neg p$ by $\bot$ in any context.
Take the program $\Pi$ consisting of the unique rule
\begin{eqnarray}
\sneg (p \wedge \neg p) \label{f:pnotp}
\end{eqnarray}
with empty body.
Interpretation $T=\{\sneg p\}$ is an answer set because $\Pi^T=\{\sneg (p \wedge \top)\}$ has $\{\sneg p\}$ as minimal model (in fact, it is the unique answer set) but if we replace $p \wedge \neg p$ by $\bot$ in $\Pi$ we get the trivial program $\{\sneg \bot\}$ whose unique answer set is $\emptyset$.
Although weak equivalence does not guarantee arbitrary replacements, it can be used to replace formulas in a theory, as stated below:
%%%%%%%%
\begin{proposition}\label{prop:replace}
Let $\alpha$, $\beta$ be a pair of formulas such that $\models \alpha \leftrightarrow \beta$.
Then, $M \models \Gamma \cup \{\alpha\}$ iff $M \models \Gamma \cup \{\beta\}$ for any theory $\Gamma$ and $\X5$-interpretation $M$.\qed
\end{proposition}
%%%%%%%%

As we mentioned before, for obtaining a congruence relation we can use validity of `$\Leftrightarrow$' instead, which guarantees the following substitution theorem.
%%%%%%%%
\begin{theorem}[Substitution]\label{th:subst}
Let $\alpha$, $\beta$ be a pair of formulas satisfying $\models \alpha \Leftrightarrow \beta$.
Then, for any formula $\varphi$, we also obtain $\models \varphi[\alpha/p] \Leftrightarrow \varphi[\beta/p]$. \qed
\end{theorem}
%%%%%%%%

Still, there are some cases in which $\leftrightarrow$ can be used for substitution, provided that the replaced formulas are not in the scope of explicit negation.

\begin{theorem}\label{th:replace}
Let $\varphi$ be a formula where atom $p$ only occurs outside the scope of explicit negation, and let $\alpha, \beta$ be two formulas satisfying $\models \alpha \leftrightarrow \beta$.
Then, $\models \varphi[\alpha/p] \leftrightarrow \varphi[\beta/p]$.\qed
\end{theorem}

%As an example, take $\models p \wedge \neg p \leftrightarrow \bot$ used before.
%
%Since this tautology does not contain explicit negation, it is also an HT tautology (Proposition~\ref{prop:httaut}).
%
%Therefore, according to Theorem~\ref{th:replace} we can use it as a pattern for other valid equivalences in $\X5$.
%
%For instance, we can replace $p$ by $\sneg(a \vee b)$ to conclude $\models \sneg (a \vee b) \wedge \neg (\sneg a \vee b) \leftrightarrow \bot$, taking $\alpha=\beta=\sneg (a \vee b)$, or to conclude $\models \sneg (a \vee b) \wedge (\sneg a \vee \sneg b) \leftrightarrow \bot$ taking $\alpha=\sneg (a \vee b)$ and $\beta=\sneg a \vee \sneg b$ since, as we will see, $\models \alpha \leftrightarrow \beta$ in $\X5$.

An important property of ASP related to HT equivalence is \emph{strong equivalence}.
We say that two programs (resp. theories) $\Gamma$ and $\Gamma'$ are \emph{strongly equivalent} iff $\Gamma \cup \Delta$ and $\Gamma' \cup \Delta$ have the same answer sets (resp. equilibrium models), for any additional program (resp. theory) $\Delta$.
When we talk about strong equivalence of formulas $\alpha$ and $\beta$ we assume they correspond to the singleton theories $\{\alpha\}$ and $\{\beta\}$.
As shown in~\cite{LPV01} (for the case without explicit negation), two programs or theories are strongly equivalent if and only if they are HT equivalent.
Since the `$\leftrightarrow$' relation in HT is congruent, there is no difference between strong equivalence (replacing formulas in a theory) and substitution (replacing subformulas in a formula).
However, as explained in~\cite{Ortiz07}, once congruence is lost, we can further refine strong equivalence in the following way.
\begin{definition}[Strong equivalence on substitution]
We say that two formulas $\alpha$ and $\beta$ are \emph{strongly equivalent on substitutions} if $\Delta \cup \{ \varphi[\alpha/p] \} $ and $\Delta \cup \{ \varphi[\beta/p] \}$ have the same equilibrium models, for any formula $\varphi$ and theory $\Delta$.
\end{definition}

The proof of the next lemma can be obtained following similar steps to the proof of the main theorem in~\cite{LPV01} replacing atoms in that case by explicit literals in ours.

\begin{lemma}\label{lem:strong.equivalence.aux}
Let $\alpha$ and $\beta$ be two formulas and be an interpretation such that $\tuple{H,T} \models\alpha$ but $\tuple{H,T} \not\models\beta$.
Then, there is a finite theory $\Delta$ such that
$\tuple{T,T} $ is an equilibrium model of one of $\Delta \cup \set{\beta}$, $\Delta \cup \set{ \alpha }$ but not of both.\qed
\end{lemma}

\begin{theorem}\label{thm:strong.equivalence}
Formulas $\alpha$ and $\beta$ are strongly equivalent iff $\models \alpha \leftrightarrow \beta$.\qed
\end{theorem}

\begin{theorem}\label{th:substeq}
Formulas $\alpha$ and $\beta$ are strongly equivalent on substitutions iff $\models \alpha \Leftrightarrow \beta$.\qed
\end{theorem}

The following set of valid equivalences allow us reducing any nested expression with explicit negation to an \emph{explicit negation normal form} (NNF) where $\sneg$ \ is only applied on atoms.
%\setlength{\multicolsep}{-25pt}
%\begin{multicols}{2}
%\begin{eqnarray}
%\sneg \top & \Leftrightarrow & \bot \label{f:nnf1}\\
%\sneg \bot & \Leftrightarrow & \top \label{f:nnf2}\\
%\sneg (\varphi \wedge \psi) & \Leftrightarrow & \sneg \varphi \,\,\vee \sneg \psi \label{f:nnf3}
%\end{eqnarray}
%\\
%\begin{eqnarray}
%\sneg (\varphi \vee \psi) & \Leftrightarrow & \sneg \varphi \,\,\wedge \sneg \psi \label{f:nnf4}\\
%\sneg \ \sneg \varphi & \Leftrightarrow & \varphi \label{f:nnf5}\\
%\sneg \neg \varphi & \Leftrightarrow & \neg \neg \varphi \label{f:nnf6}
%\end{eqnarray}
%\end{multicols}
%\vspace{30pt}
%
\begin{eqnarray}
\sneg \top & \Leftrightarrow & \bot \label{f:nnf1}\\
\sneg \bot & \Leftrightarrow & \top \label{f:nnf2}\\
\sneg (\varphi \wedge \psi) & \Leftrightarrow & \sneg \varphi \,\,\vee \sneg \psi \label{f:nnf3}\\
\sneg (\varphi \vee \psi) & \Leftrightarrow & \sneg \varphi \,\,\wedge \sneg \psi \label{f:nnf4}\\
\sneg \ \sneg \varphi & \Leftrightarrow & \varphi \label{f:nnf5}\\
\sneg \neg \varphi & \Leftrightarrow & \neg \neg \varphi \label{f:nnf6}
\end{eqnarray}
For instance, we can reduce the nested expression~\eqref{f:pnotp} to NNF as follows:
\[
\begin{array}{rcll}
\sneg (p \wedge \neg p) & \Leftrightarrow & \sneg p \vee \sneg \neg p & \mbox{ by } \eqref{f:nnf3}\\
& \Leftrightarrow & \sneg p \vee \neg \neg p & \mbox{ by } \eqref{f:nnf6}
\end{array}
\]
Programs in NNF correspond to the original syntax in~\cite{LTT99}.
That paper provided several transformations that allowed reducing any program in NNF to a regular program.
These transformations included commutativity and associativity of conjunction and disjunction (which are obviously satisfied in $\X5$) plus the equivalences in the following proposition.

\begin{proposition}
The following formulas are $\X5$ tautologies:
\begin{eqnarray}
%\varphi \wedge \psi \Leftrightarrow \psi \wedge \varphi & & 
%\varphi \vee \psi \Leftrightarrow \psi \vee \varphi 
%\label{f:comm} \\
%\varphi \wedge (\psi \wedge \gamma) \Leftrightarrow (\varphi \wedge \psi) \wedge \gamma & & 
%\varphi \vee (\psi \vee \gamma) \Leftrightarrow (\varphi \vee \psi) \vee  \gamma
%\label{f:assoc} \\
\varphi \wedge (\psi \vee \gamma) \Leftrightarrow (\varphi \wedge \psi) \vee (\varphi \wedge \gamma) & & 
\varphi \vee (\psi \wedge \gamma) \Leftrightarrow (\varphi \vee \psi) \wedge (\varphi \vee \gamma)
\label{f:distrib} \\
\varphi \wedge \bot \Leftrightarrow \bot & &
\varphi \vee \top \Leftrightarrow \top 
\label{f:anhil} \\
\varphi \wedge \top \Leftrightarrow \varphi & & 
\varphi \vee \bot \Leftrightarrow \varphi 
\label{f:neut} \\
\neg (\varphi \wedge \psi) \Leftrightarrow \neg \varphi \vee \neg \psi & &
\neg (\varphi \vee \psi) \Leftrightarrow \neg \varphi \wedge \neg \psi
\label{f:demorgan} \\
\neg \top \Leftrightarrow \bot & & 
\neg \bot \Leftrightarrow \top \label{f:notconst}
\end{eqnarray}
\vspace{-20pt}
\begin{eqnarray}
\neg \neg \neg \varphi & \Leftrightarrow & \neg \varphi \label{f:triplenot}\\
\varphi \to \psi \wedge \gamma & \Leftrightarrow & (\varphi \to \psi) \wedge (\varphi \to \gamma) \label{f:andhead}\\
\varphi \vee \psi \to \gamma & \Leftrightarrow & (\varphi \to \gamma) \wedge (\psi \to \gamma) \label{f:orbody}\\
\varphi \wedge \neg \neg \psi \to \gamma & \Leftrightarrow & \varphi \to \gamma \vee \neg \psi \label{f:notbody}\\
\varphi \to \gamma \vee \neg \neg \psi & \Leftrightarrow & \varphi \wedge \neg \psi \to \gamma \label{f:nothead}
\end{eqnarray}
and correspond to the transformations in~\cite{LTT99}.\qed
\end{proposition}

For instance, as we saw, \eqref{f:pnotp} was equivalent to $\sneg p \vee \neg \neg p$ but this can be further transformed into the regular rule $\neg p \to \sneg p$ commonly used to assign falsity of $p$ by default.
\begin{example}[Example~\ref{ex:bird} continued]
Rule \eqref{f:bird} can be transformed as follows:
\[
\begin{array}{rcl@{\ \ \ }l}
\eqref{f:bird}  & \Leftrightarrow & 
\neg \I{bird} \vee \neg\!\!\sneg \I{flies} \to \ \sneg (\I{bird} \wedge \sneg \I{flies})
& \mbox{by } \eqref{f:demorgan}\\
& \Leftrightarrow & 
\neg \I{bird} \vee \neg\!\!\sneg \I{flies} \to \ \sneg \I{bird} \vee \sneg \ \sneg \I{flies}
& \mbox{by } \eqref{f:nnf3} \\
& \Leftrightarrow & 
\neg \I{bird} \vee \neg\!\!\sneg \I{flies} \to \ \sneg \I{bird} \vee \I{flies}
& \mbox{by } \eqref{f:nnf5} \\
& \Leftrightarrow & 
(\neg \I{bird} \to \ \sneg \I{bird} \vee \I{flies})\\
&  & 
\wedge (\neg\!\!\sneg \I{flies} \to \ \sneg \I{bird} \vee \I{flies})
& \mbox{by } \eqref{f:orbody} \\
\end{array}
\]
and the last step is a conjunction of two regular rules as in standard ASP solvers.\qed
\end{example}

Reduction to NNF is also possible on arbitrary formulas.
For that purpose, we can combine \eqref{f:nnf1}-\eqref{f:nnf6} with the following valid (weak) equivalence:
\begin{eqnarray}
\sneg (\varphi \to \psi) & \leftrightarrow & \neg \neg \varphi \wedge \sneg \psi \label{f:nnf7}
\end{eqnarray}
However, the reduction must be done with some care, because this last equivalence cannot be shifted to $\Leftrightarrow$.
Indeed, the left and right expressions have different valuations when $M(\varphi)=M(\psi)=1$, obtaining $M(\sneg (\varphi \to \psi))=-2 \neq -1 = M(\neg \neg \varphi \wedge \sneg \psi)$.
Fortunately, Theorem~\ref{th:replace} allows us applying \eqref{f:nnf7} from the outermost occurrence of $\sim$ and then recursively combining with \eqref{f:nnf1}-\eqref{f:nnf6} until $\sim$ is only applied to atoms.
\begin{theorem}
For any formula $\varphi$ there exists a formula $\psi$ in NNF such that $\models \varphi \leftrightarrow \psi$.\qed
\end{theorem}

For instance, we can reduce the following formula into NNF as follows:
\begin{eqnarray*}
\sim (a \to \ \sneg b \wedge (c \to d)) & \leftrightarrow & \neg \neg a \wedge \sim( \sneg b \wedge (c \to d)) \\
& \leftrightarrow & \neg \neg a \wedge (\sneg \ \sneg b \vee \sneg (c \to d)) \\
& \leftrightarrow & \neg \neg a \wedge (b \vee \neg \neg c \wedge \sneg d)
\end{eqnarray*}
However, we cannot apply~\eqref{f:nnf7} making a replacement in the scope of explicit negation.
A clear counterexample is the formula $\sneg \ \sneg (p \to q)$ that, due to \eqref{f:nnf5}, is strongly equivalent to $p \to q$, but applying \eqref{f:nnf7} inside would incorrectly lead to the nested expression $\sneg (\neg \neg p \wedge \sneg q)$ that can be transformed into the strongly equivalent expression $\neg p \vee q$, different from $p \to q$ in ASP.

%%%%%%%%%%%%%%%%%%%%%%%%%%%%%%%%%%%%%%%%%%%%%%%%%%%%%%%%%%%
\section{Related work}\label{sec:related}

As explained in the introduction, this work is obviously related to the characterisation of `$\sim$' as Nelson's \emph{strong negation}~\cite{Nel49} for intermediate logics.
In particular, the addition of strong negation to HT produces the five-valued logic $\N5$ already present in the original definition of Equilibrium Logic~\cite{Pearce96}.
In fact, the interpretations and the truth values we have chosen for $\X5$ coincide with those for $\N5$, and their evaluation of (non-derived) connectives $\top, \wedge, \vee$ and $\to$ from Figure~\ref{fig:tables} also coincide in both logics, except for one difference in the table of implication: the value for $M(\varphi)=1$ and $M(\psi)=-2$ changes from $-2$ to $-1$ in $\N5$.
This change and its result on derived operators is shown in Figure~\ref{fig:tablesn5} where the different values are framed in rectangles.
\def\minusone{\boxed{\!-1\!\!\!}}
\begin{figure}[htbp]
\centering
$$
\begin{array}{c@{\hspace{20pt}}c}
\begin{array}{r|rrrrr}
\to & -2 & -1 & 0 & 1 & 2\\
\hline
-2 & 2 & 2 & 2 & 2 & 2 \\
-1 & 2 & 2 & 2 & 2 & 2 \\
0 & 2 & 2 & 2 & 2 & 2 \\
1 & \minusone & -1 & 0 & 2 & 2 \\
2 & -2 & -1 & 0 & 1 & 2
\end{array}
&
\begin{array}{c@{\hspace{20pt}}c}
\begin{array}{r|r}
\varphi & \neg \varphi\\
\hline
-2 & 2 \\
-1 & 2 \\
0 & 2 \\
1 & \minusone \\
2 & -2 
\end{array}
\end{array}
\\ \\
\begin{array}{r|rrrrr}
\leftrightarrow & -2 & -1 & 0 & 1 & 2\\
\hline
-2 & 2 & 2 & 2 & \minusone & -2 \\
-1 & 2 & 2 & 2 & -1 & -1 \\
0 & 2 & 2 & 2 & 0 & 0 \\
1 & \minusone & -1 & 0 & 2 & 1 \\
2 & -2 & -1 & 0 & 1 & 2
\end{array}
&
\begin{array}{r|rrrrr}
\Leftrightarrow & -2 & -1 & 0 & 1 & 2\\
\hline
-2 & 2 & 1 & 0 & \minusone & -2 \\
-1 & 1 & 2 & 0 & -1 & -2 \\
0 & 0 & 0 & 2 & 0 & 0 \\
1 & \minusone & -1 & 0 & 2 & 1 \\
2 & -2 & -2 & 0 & 1 & 2
\end{array}
\end{array}
$$
\caption{Truth tables for $\N5$ that differ from $\X5$.}
\label{fig:tablesn5}
\end{figure}
As a result, $\N5$ ceases to satisfy \eqref{f:nnf6} and \eqref{f:nnf7} whose role in the reduction to NNF is respectively replaced by the $\N5$-valid weak equivalences:
\begin{eqnarray}
\sneg \neg \varphi & \leftrightarrow & \varphi \label{f:N1}\\
\sneg (\varphi \to \psi) & \leftrightarrow & \varphi \wedge \sneg \psi \label{f:N2}
\end{eqnarray}
The difference between \eqref{f:nnf7} and \eqref{f:N2} also reveals the effect on falsification of implication in both logics.
While $\tuple{H,T} \falsif \varphi \to \psi$ requires $\tuple{T,T} \models \varphi$ in $\X5$, this is replaced by condition $\tuple{H,T}\models \varphi$ in $\N5$.
Curiously, although these two logics provide a different behaviour for $\sneg$ \ as strong versus explicit negation, they actually have the same evaluation for that connective, while their real technical difference lies on falsity of implication.

The reason why $\N5$ does not capture the extended reduct for nested expressions proposed in this paper is that \eqref{f:triplenot} is not valid in that logic.
This is because, when $M(\varphi)=1$, we get $M(\neg \varphi)=-1 \neq -2 = M(\neg \neg \neg \varphi)$.
It is still possible to define $\N5$ operators in $\X5$ as follows:
\begin{eqnarray*}
\varphi \stackrel{\N5}{\to} \psi & \eqdef & \varphi \to \ \sneg \varphi \vee \psi \\
\stackrel{\N5}{\neg} \varphi & \eqdef & \varphi \to \ \sneg \varphi
\end{eqnarray*}
using here the $\X5$ interpretation for implication.
Analogously, we can also define the $\X5$ operators in $\N5$ in the following way:
\begin{eqnarray*}
\varphi \stackrel{\X5}{\to} \psi & \eqdef & (\varphi \to \psi) \wedge (\sneg \psi \to \ \neg \neg \neg \varphi)  \\
\stackrel{\X5}{\neg} \varphi & \eqdef & \neg \neg \neg \varphi
\end{eqnarray*}
assuming that we interpret implication and $\neg$ under $\N5$ instead.

An interesting connection between both variants is that the addition of the excluded middle axiom schemata $\varphi \vee \neg \varphi$ imposes the restriction of total models $\tuple{T,T}$ both in $\X5$ and in $\N5$.
This means that all atoms and formulas are evaluated in the set $\{-2,0,2\}$, for which the truth tables coincide in these two logics and actually collapse to classical logic with strong negation~\cite{vakarelov1977notes} introduced in Section~\ref{sec:nested}.
This coincidence is important since equilibrium models (and so, answer sets) are total models.

To conclude the section on related work, another possibility for interpreting a second negation `$\sneg$' inside intuitionistic logic was provided by~\cite{FH96} using a \emph{classical} negation interpretation.
Although the idea seems closer to Gelfond and Lifschitz' original terminology for a second negation, it actually provides undesired effects from an ASP point of view.
Classical negation in HT means keeping only the satisfaction relation `$\models$' in Definition~\ref{def:satfals} (falsification `$\falsif$' is not needed) but replacing the condition for `$\sneg$' so that $\tuple{H,T} \models \sneg \varphi$ if $\tuple{H,T} \not\models \varphi$.
One important effect of this change is that HT with classical negation ceases to satisfy the persistence property (Theorem~\ref{th:persistence}).
But perhaps a more important problem from the ASP perspective is that $\neg p$ implies $\sneg p$ for any atom $p$.
Thus, the rule $\neg p \to \sneg p$ becomes a tautology in this context, whereas it is normally used in ASP to conclude that $p$ is explicitly false by default.

%%%%%%%%%%%%%%%%%%%%%%%%%%%%%%%%%%%%%%%%%%%%%%%%%%%%%%%%%%%
\section{Conclusions}
\label{sec:conc}

We have introduced a variant of constructive negation in Equilibrium Logic (and its monotonic basis, HT) we called \emph{explicit negation}.
This variant shares some similarities with the previous formalisation based on Nelson's strong negation, but changes the interpretation for falsity of implication.
We have also introduced a reduct-based definition of answer sets for programs with nested expressions extended with explicit negation, proving the correspondence with equilibrium models.

For future work, we will study a possible axiomatisation.
To this aim, it is interesting to observe that the formulas \eqref{f:nnf3}-\eqref{f:nnf5} (in their weak equivalence versions) plus  \eqref{f:N1} and \eqref{f:N2} actually correspond to Vorob'ev axiomatisation~\cite{Vo52a,Vo52b} of strong negation in intuitionistic logic.
As we saw, the role of \eqref{f:N1} and \eqref{f:N2} in $\N5$ is replaced in $\X5$ by \eqref{f:nnf5} and \eqref{f:nnf7}, so an interesting question is whether this replacement may become a complete axiomatisation for explicit negation in $\X5$ or intuitionistic logic in the general case.
We also plan to explore the effect of explicit negation on extensions of equilibrium logic, revisiting the use of strong negation in paraconsistent~\cite{OdintsovP05} and partial~\cite{COP06} equilibrium logic, or considering its combination with partial functions~\cite{Cab11,CabalarCPV14}, and temporal~\cite{ACD+13} or epistemic~\cite{CerroHS15,CFF19} reasoning.

\bibliographystyle{acmtrans}
\bibliography{refs}
\label{lastpage}

\end{document}

%% file: letter.tex
%!TEX root = main.tex

\newpage
\newpage
\section*{Summary of changes}

\pagestyle{empty}

We wish to thank the reviewers for their thorough evaluation and their useful comments and suggestions.
Each comment is answered below using the code R$x$.$yy$ for comment number $yy$ from reviewer $x$.
To facilitate the reviewing process, each change is colored in red and makes a marginal reference to the reviewer's comment code R$x$.$yy$.

\subsection*{\bf Reviewer 1}
\begin{enumerate}[label=R1.\arabic*]
\item\label{rev1.1} \rev{It's puzzling that that the authors don't talk at all about the definition of the reduct due to Ferraris.  It exactly corresponds to equilibrium logic limited to formulas with one negation. Is there a way to add the second negation to Ferraris' definition to achieve exact correspondence with the new version of equilibrium logic?}

Thank you for the suggestion. We have defined now an extension of Ferraris' reduct (see Definition~\ref{def:Ferraris_reduct}) that is applicable to arbitrary theories with explicit negation and precisely captures the semantics of $\X5$ (Theorem~\ref{th:Ferraris_reduct}). 

\end{enumerate}

%%%%%%%%%%%%%%%%%%%%%%%%%%%%%%%%%%%%%%%%%%%%%%%%%%%%%%%%%%%%%%%%%%%%%%%%
\subsection*{\bf Reviewer 2}
\begin{enumerate}[label=R2.\arabic*]
\item\label{rev2.1}  \rev{ Perhaps just a comment could be added about the practical relevance of the
 contribution and which are the contexts where such extended nested
 expressions could be applied.}
 See answer to Reviewer 3 below (\ref{rev3.1})
\end{enumerate}

%%%%%%%%%%%%%%%%%%%%%%%%%%%%%%%%%%%%%%%%%%%%%%%%%%%%%%%%%%%%%%%%%%%%%%%%
\subsection*{\bf Reviewer 3}
\begin{enumerate}[label=R3.\arabic*]
\item\label{rev3.1} \rev{It was not clear to me what is the motivation for this work. What is the problem we try to solve here, and why is it important? How are we supposed to construct the extensions of programs in this new logic?}

As we explain in the introduction, the use of explicit negation in ASP has been a quite common feature since the very beginning of this paradigm. We have included a pair of new paragraphs in the introduction dealing with a well-known example (trains crossing) in the literature. This example illustrates the difference between default and explicit negation from a Knowledge Representation point of view. Normally, explicit negation is used as a prefix on atoms and is not treated as a real operator. This doesn't happen with conjunction, disjunction and default negation, whose nesting was generalised in~\cite{LifschitzTT99} so, to put an example, several rules sharing the same consequent $A \to B$, $A \to C$ can be abbreviated with a conjunction in the head $A \to B \wedge C$. In this paper, we share the same motivation and look for a general treatment of explicit negation so it becomes a real operator. In the new paragraph, we have also included an example of its potential use when allowing arbitrary nesting.

\end{enumerate}